\documentclass[onecolumn,amsmath,amssymb,10pt,aps]{revtex4}
  
\usepackage[utf8]{inputenc}
\usepackage[T1]{fontenc}

\usepackage{amsmath}
\usepackage{amsfonts}
\usepackage{graphicx}
\usepackage{yfonts}
\usepackage{color}
\usepackage[normalem]{ulem}
\usepackage{amsthm}
\usepackage{bm}
\usepackage{bbm}
\usepackage{mathtools}
\usepackage{array}
\usepackage{placeins}
\usepackage{enumitem}

\newcommand\bigforall{\mbox{\Large $\mathsurround=1pt\forall$}}

\def\<{\langle}
\def\>{\rangle}
\newcommand{\tr}{\mathrm{Tr}}
\newcommand{\Tr}{\mathrm{Tr}}
\def\oper{{\mathchoice{\rm 1\mskip-4mu l}{\rm 1\mskip-4mu l}
{\rm 1\mskip-4.5mu l}{\rm 1\mskip-5mu l}}}
\DeclareMathAlphabet\mathbfcal{OMS}{cmsy}{b}{n}

\mathchardef\mhyphen="2D 

\newtheorem{Theorem}{Theorem}

\newtheorem{Definition}{Definition}
\newtheorem{Remark}{Remark}
\newtheorem{Proposition}{Proposition}
\newtheorem{Example}{Example}

\begin{document}

\title{Equiangular weighted frames and conical 2-designs with unequal traces}

\author{Katarzyna Siudzi\'{n}ska\footnote{e-mail: kasias@umk.pl}}
\affiliation{Institute of Physics, Faculty of Physics, Astronomy and Informatics, Nicolaus Copernicus University in Toru\'{n}, ul.~Grudzi\k{a}dzka 5, 87--100 Toru\'{n}, Poland}

\begin{abstract}
Conical 2-designs are families of positive operators possessing crucial symmetry properties that allow them to efficiently characterize important quantum measures. However, little is known about their general structures and properties. We propose a construction method of conical 2-designs whose elements are not of equal traces, based on a generalization of equiangular weighted frames to non-projective operators. We provide a detailed analysis of their properties, relations to quantum measurements, and applications, together with examples that go beyond the known classes. This constitutes a major step toward a full characterization of informationally overcomplete quantum measurements.
\end{abstract}

\flushbottom

\maketitle

\thispagestyle{empty}

\section{Introduction}

Recently, much attention has been given to positive, operator-valued measures (POVMs), which describe the measurement process in quantum information theory. They are represented by collections of positive, semi-definite operators on the Hilbert space that provide a resolution of the identity operator. Special interest is focused on symmetric classes of POVMs, whose properties allow them to find important applications e.g. in entanglement detection \cite{SICMUB_preprint,multi_SM,Rastegin_EW,SM_Pmaps}, quantum steering \cite{Certify_EPR,Certify_EPR2,steering_entropic}, and quantum tomography \cite{Nguyen,Innocenti}. On one hand, there are equioverlapping measurements \cite{EOM22,EOM24,EOM25,EOMq3}, which generalize symmetric, informationally complete (SIC) POVMs \cite{Renes} and semi-SIC POVMs \cite{semi-SIC} to projective operators of distinct traces. On the other, families of mutually unbiased equiangular measurements \cite{GEAM,SIC-MUB_general,SIC-MUB} introduce informationally overcomplete POVMs that unify the framework for SIC POVMs and mutually unbiased bases (MUBs) \cite{Schwinger,Szarek}. The motivation behind these constructions is to characterize quantum measurements that perform better under certain tasks than SIC POVMs and MUBs, while also replacing them in the dimensions where they cannot be constructed.

A different approach to classify symmetric semi-positive operators has been proposed in ref. \cite{Graydon}. There, the authors introduce conical 2-designs as a generalization of complex projective 2-designs to semi-positive operators of arbitrary rank. Their recent applications include entropic uncertainty relations \cite{c2d_uncertainty,c2d_uncertainty2} and quantum coherence \cite{c2d_coherence}. While both SIC POVMs and MUBs are conical (and projective) 2-designs, this is not true for all generalized equiangular measurements. It has been shown that only equidistant GEAMs form conical 2-design \cite{GEAM}. Interestingly, only for GEAMs that are conical 2-designs, only two real parameters are needed to fully characterize important quantum measures, like quantum coherence, quantum concurrence, and quantum entanglement \cite{GEAM_coherence}. Equivalence relations between classes of generalized equiangular measurements and conical 2-designs have been examined in ref. \cite{conical}. It turns out that there exist even more general conical 2-designs, as in general all operators can have distinct traces.

Here, we analyze a class of conical 2-designs that goes beyond generalized equiangular measurements and mutually unbiased generalized equiangular tight frames (MU GETFs). Namely, we drop the condition of the frames being tight and instead consider generalized weighted frames with weights assigned to operator traces. This allows us to define a wide class of mutually unbiased (MU) generalized equiangular weighted frames (GEWFs), where all operators can possess distinct traces. We analyze their properties, like linear independence or informational completeness, and present a general construction method. Our discussion is illustrated with low-dimensional examples that go beyond the known operator classes. We then show how MU GEWFs relate to POVMs, conical 2-designs, and MU GETFs. We end the discussion with applications in the index of coincidence, variance, and entanglement detection.

The paper is organized as follows. In Section 2, we introduce a generalization of weighted tight frames (GEWFs) from rank-1 projectors to positive, semi-definite linear operators. Assuming that the weights are given by operator traces, we derive their trace relations and show that all GEWF elements are linearly independent. Next, we propose four general methods of their construction from an orthonormal Hermitian operator basis. In Section 3, we extend the notion of GEWFs to mutually unbiased (MU) collections of GEWFs. We derive conditions under which MU GEWFs form an informationally overcomplete set. We analyze the multiple ways to construct them from an orthonormal Hermitian operator basis, together with the relations between the operators from different constructions. Section 4 shows how MU GEWFs relate to MU GETFs and generalized equiangular measurements. Next, Section 5 establishes the conditions under which MU GEWFs are conical 2-designs. Together with the results from ref. \cite{conical}, it also formulates an equivalence relation between subclasses of MU GEWFs and conical 2-designs. In Section 6, we show how to use any MU GEWFs to define a POVM. Finally, Section 7 lists important applications of MU GEWFs. We prove that the index of coincidence, variance, and Schmidt number criterion can be characterized using two positive constants that define a conical 2-design. The most important results and open questions are then summarized in Section 8.

\section{Equiangular weighted frames}

Consider the Hilbert space $\mathcal{H}\simeq\mathbb{C}^d$ associated with a $d$-dimensional (qudit) quantum system. In the space of operators on $\mathbb{C}^d$, pure quantum states are represented via rank-1 projectors $|\psi_k\>\<\psi_k|$ onto the normalized vectors $\psi_k\in\mathbb{C}^d$. If we allow for the multiplication of $|\psi_k\>\<\psi_k|$ by non-zero real factors $\alpha_k$, then each rank-1 projector turns into a line $\{R_k=\alpha_k|\psi_k\>\<\psi_k|;\,\alpha_k\in\mathbb{R}/\{0\}$. Much interest has been given to the collections $\{R_k;\,k=1,\ldots,M\}$ of lines that are equiangular, which means they satisfy the condition $\Tr(R_kR_\ell)=y\Tr(R_k)\Tr(R_\ell)$ for any $k\neq\ell$. However, unlike in ref. \cite{GEAM}, we drop the requirement for $R_k$ to be uniform, allowing for $\alpha_k\neq \alpha_\ell$. Now, equiangular lines such that $\sum_{k=1}^Mw_kR_k=\gamma\mathbb{I}_d$, $\gamma>0$, define a weighted tight frame \cite{weighted,weighted2}. In further considerations, we limit ourselves to weights given by operator traces, $w_k=\Tr(R_k)$. In addition, we drop the requirement that $R_k$ are rank-1, so that they are instead positive semi-definite operators of arbitrary rank. Let us summarize this construction in a definition.

\begin{Definition}
The generalized equiangular weighted frame (GEWF) is a set $\mathcal{R}=\{R_k:\,k=1,\ldots,M\}$ of $M$ semi-positive operators $R_k$ on $\mathbb{C}^d$, $2\leq M\leq d^2$, such that
\begin{equation}\label{def}
\begin{split}
\Tr(R_k)&=w_k,\\
\Tr(R_k^2)&=x_k,\\
\Tr(R_kR_\ell)&=yw_kw_\ell,\quad \ell\neq k,\\
\sum_{k=1}^Mw_kR_k&=\gamma\mathbb{I}_d,\quad \gamma>0.
\end{split}
\end{equation}
\end{Definition}

From this definition, it follows that $w_k>0$,
\begin{equation}\label{deff}
x_k=c+\frac{\gamma-c}{d\gamma}w_k^2,\qquad y=\frac{\gamma-c}{d\gamma},
\end{equation}
and
\begin{equation}\label{war1}
0<c\leq\gamma=\frac 1d \sum_{k=1}^M w_k^2,
\end{equation}
where $c=0$ is excluded because then $R_k$ is proportional to the identity operator. The upper bound $c=\gamma$ is reached only when $w_k=1$, for which $\gamma=M/d$ and $M\leq d$. Additionally, from $\Tr(R_k^2)\leq[\Tr(R_k)]^2$, one has
\begin{equation}\label{war2}
w_k\geq\sqrt{\frac{d\gamma c}{(d-1)\gamma+c}},
\end{equation}
which also implies an additional constraint on the parameter $c$,
\begin{equation}\label{c2}
c\leq \frac{d-1}{M-1}\gamma.
\end{equation}
Note that the formula for $\Tr(R_k^2)$ always contains a free term that is not multiplied by $[\Tr(R_k)]^2$. 
Moreover, the elements of $\mathcal{R}$ are equiangular, but they form a tight frame if and only if $w_k=w_\ell$ for all $k\neq\ell$.

\begin{Proposition}
The operators $R_k$, $k=1,\ldots,M\leq d^2$, are linearly independent.
\end{Proposition}

\begin{proof}
Let us introduce an auxiliary operator
\begin{equation}
R=\sum_{k=1}^Mr_kR_k,\qquad r_k\in\mathbb{R},
\end{equation}
and assume that $R=0$. From eq. (\ref{def}), we find that
\begin{equation}
\Tr(R)=\sum_{k=1}^Mr_kw_k=0,\qquad \Tr(RR_\ell)=\sum_{k=1}^Mr_k\frac{\gamma-c}{d\gamma}w_kw_\ell+cr_\ell=0.
\end{equation}
Therefore, $\sum_{k=1}^Mr_kR_k=0$ implies that all $r_k=0$,
which means that $R_k$ are indeed linearly independent.
\end{proof}

From the above proposition, we see that if $M=d^2$, then the GEWF elements $R_k$, $k=1,\ldots,M$, form an informationally complete set.

Now, to construct $R_k$ in any dimension $d$, take a subset $\mathcal{G}=\{G_k:\,k=1,\ldots,M\}$ of $M$ operators $G_k$ from an orthonormal Hermitian operator basis $\{G_k:\,k=0,\ldots,d^2-1\}$ with $G_0=\mathbb{I}_d/\sqrt{d}$ and traceless $G_k$, $k=1,\ldots,d^2-1$. Then, one has two possible results: $\{R_k\}=\{R_k^{(+)}\}$ or $\{R_k\}=\{R_k^{(-)}\}$, where
\begin{equation}\label{gen}
R_k^{(\pm)}=\frac{w_k}{d}\left(\mathbb{I}_d\pm\sqrt{\frac{dc}{\gamma}}H_k\right).
\end{equation}
The operators $H_k$ are linear combinations of $G_\ell$. It remains to find the expansion coefficients.

\begin{Proposition}\label{4Hk}
The traceless operators $H_k$ in eq. (\ref{gen}) are given by
\begin{equation}
(i)\quad\left\{
\begin{split}
H_k&=\frac{\sqrt{d\gamma}}{w_k}G_k-\sum_{\ell=1}^{M-1}\frac{w_\ell}{\sqrt{d\gamma}-w_M} G_\ell,\\
H_M&=\sum_{\ell=1}^{M-1}\frac{w_\ell}{w_M} G_\ell,
\end{split}
\right.\qquad{\rm or}\qquad(ii)\quad
\left\{
\begin{split}
H_k^\prime&=\frac{\sqrt{d\gamma}}{w_k}G_k-\sum_{\ell=1}^{M-1}\frac{w_\ell}{\sqrt{d\gamma}+w_M} G_\ell,\\
H_M^\prime&=-\sum_{\ell=1}^{M-1}\frac{w_\ell}{w_M} G_\ell,
\end{split}
\right.
\end{equation}
where $k=1,\ldots,M-1$.
\end{Proposition}

\begin{proof}
Assume that $H_k$ are linear combinations of $G_\ell$ with the following choice of real coefficients,
\begin{equation}\label{Hkgen}
\begin{split}
H_k&=\xi_kG_k-\sum_{\ell=1}^{M-1}\mu_\ell G_\ell,\qquad k=1,\ldots,M-1,\\
H_M&=\sum_{\ell=1}^{M-1}\nu_\ell G_\ell.
\end{split}
\end{equation}
This is not the most general choice, and hence there may exist other constructions of $H_k$. From eqs. (\ref{def}) and (\ref{gen}), we have
\begin{equation}\label{cond2}
\Tr(H_k^2)=\frac{d\gamma}{w_k^2}-1,\qquad
\Tr(H_kH_\ell)=-1.
\end{equation}
Moreover, the condition $\sum_{k=1}^Mw_kR_k=\gamma\mathbb{I}_d$ reduces to
\begin{equation}\label{cond1}
\sum_{k=1}^{M}w_k^2H_k=0.
\end{equation}
It remains to choose the coefficients $\xi_k$, $\mu_k$, and $\nu_k$ in such a way that eqs. (\ref{cond2}) and (\ref{cond1}) always hold. Starting from eq. (\ref{cond1}), we immediately find that
\begin{equation}
\sum_{k=1}^{M-1}G_k\Big[-\mu_k(d\gamma-w_M^2)+w_k^2\xi_k+w_M^2\nu_k\Big]=0,
\end{equation}
and hence, due to linear independence of $G_k$,
\begin{equation}
\nu_k=\frac{1}{w_M^2}\Big[\mu_k(d\gamma-w_M^2)-w_k^2\xi_k\Big].
\end{equation}
Next, eq. (\ref{cond2}) impose the following constraints,
\begin{equation}\label{two}
\begin{split}
\frac{d\gamma}{w_k^2}-1&=\xi_k^2-2\mu_k\xi_k+\sum_{m=1}^{M-1}\mu_m^2,\\
-1&=-\mu_k\xi_k-\mu_\ell\xi_\ell+\sum_{m=1}^{M-1}\mu_m^2,
\end{split}
\end{equation}
with $k,\ell=1,\ldots,M-1$ and $\ell\neq k$. After subtracting these two equations, we see that
\begin{equation}\label{three}
\bigforall_{\ell\neq k}\quad \xi_k(\xi_k-\mu_k)+\mu_\ell\xi_\ell=\frac{d\gamma}{w_k^2}.
\end{equation}
Therefore, $\mu_\ell\xi_\ell$ does not depend on the choice of $\ell$. If we apply this to eq. (\ref{three}) and the second line in eq. (\ref{two}), it becomes clear that our problem reduces to solving the system of three equations:
\begin{align}
A&=\mu_k\xi_k,\label{eq1}\\
\sum_{m=1}^{M-1}\mu_m^2&=2A-1,\label{eq2}\\
\xi_k(\xi_k-\mu_k)&=\frac{d\gamma}{w_k^2}-A.\label{eq3}
\end{align}
The solutions read
\begin{equation}\label{ximu}
\xi_k=\pm\frac{\sqrt{d\gamma}}{w_k},\qquad \mu_k=\pm\frac{Aw_k}{\sqrt{d\gamma}},
\end{equation}
where $A$ is satisfies the quadratic equation
\begin{equation}\label{eqA}
(d\gamma-w_M^2)A^2-2d\gamma A+d\gamma=0\qquad\Longleftrightarrow\qquad
A=\frac{\sqrt{d\gamma}}{\sqrt{d\gamma}\pm w_M}.
\end{equation}
In what follows, we present the derivations only for the positive solutions for $\xi_k$ and $\mu_k$. The negative solutions are achieved by simply switching the sign of the traceless basis operators ($G_k\to-G_k$), and hence they merely swap $R_k^{(+)}\leftrightarrow R_k^{(-)}$. Therefore, eq. (\ref{ximu}) provides us with one solution for $\xi_k$, as well as two solutions for $\mu_k$ and $\nu_k$, one for each value of $A$:
\begin{equation}\label{ximu2}
\xi_k=\frac{\sqrt{d\gamma}}{w_k},\qquad \mu_k=\frac{w_k}{\sqrt{d\gamma}\pm w_M},\qquad \nu_k=\pm\frac{w_k}{w_M}.
\end{equation}
\end{proof}

The results of Proposition \ref{4Hk} show that a single operator basis produces four informationally complete GEWFs: $\mathcal{R}_\pm=\{R_k^{(\pm)}:\,k=1,\ldots,M\}$ and $\mathcal{R}^\prime_\pm=\{R_k^{\prime(\pm)}:\,k=1,\ldots,M\}$, with elements given by eq. (\ref{gen}) and
\begin{equation}\label{genprime}
R_k^{\prime(\pm)}=\frac{w_k}{d}\left(\mathbb{I}_d\pm\sqrt{\frac{dc}{\gamma}}H_k^\prime\right).
\end{equation}
Note that the constant $c$ does not appear anywhere in the formulas for $H_k$ and $H_k^\prime$. Therefore, it is only responsible for controling the semi-positivity of the GEWF operators. Notably, its value can be different depending on the construction method. Indeed, let us denote two bounds for $c$ by
\begin{equation}\label{cc}
c_+=\frac{\gamma}{d\lambda_{\min}^2},\qquad c_-=\frac{\gamma}{d\lambda_{\max}^2},
\end{equation}
where $\lambda_{\min}$ and $\lambda_{\max}$ are the minimal and maximal eigenvalues from among all eigenvalues of $H_k$ or $H_k^\prime$, $k=1,\ldots,M$. 
Now, for $R_k^{(+)}$ and $R_k^{\prime(+)}$, the semi-positivity condition is equivalent to $0<c\leq c_+$. Analogically, for $R_k^{(-)}$ and $R_k^{\prime(-)}$, one has $0<c\leq c_-$. In general, $c_+\neq c_-$, and below we present an example where $c=c_+=c_-$ for both $R_k^{(\pm)}$ and $R_k^{\prime(\pm)}$.

\begin{Example}\label{excon}
Consider a qubit case ($d=2$) with $\gamma=25$ and $M=4$. Take the Hermitian orthonormal basis of normalized Pauli matrices $G_k=\sigma_k/\sqrt{2}$, $k=1,2,3$. The operator traces are chosen as follows: $w_1=w_2=3$, $w_3=w_4=4$. Depending on the construction method, one obtains the following GEWFs:
\begin{enumerate}[label=(\it\roman*)]
\item $ $
\begin{equation}
\begin{split}
R_1^{(+)}&=\frac{3}{34\sqrt{82}}
\begin{pmatrix}
-60-\xi_- & \nu_-+i\eta_+ \\
\nu_--i\eta_+ & 60+\xi_+
\end{pmatrix},\\
R_2^{(+)}&=\frac{3}{34\sqrt{82}}
\begin{pmatrix}
-60-\xi_- & -\eta_+-i\nu_- \\
-\eta_++i\nu_- & 60+\xi_+
\end{pmatrix},
\end{split}\qquad
\begin{split}
R_3^{(+)}&=\frac{1}{17\sqrt{82}}
\begin{pmatrix}
135-2\xi_- & -2(1-i)\eta_+ \\
-2(1+i)\eta_+ & -135+2\xi_+
\end{pmatrix},\\
R_4^{(+)}&=\frac{1}{2\sqrt{41}}
\begin{pmatrix}
4(\sqrt{41}+3) & 9(1-i) \\
9(1+i) & 4(\sqrt{41}-3)
\end{pmatrix},
\end{split}
\end{equation}
\item $ $
\begin{equation}
\begin{split}
R_1^{(-)}&=\frac{3}{34\sqrt{82}}
\begin{pmatrix}
60+\xi_+ & -\nu_--i\eta_+ \\
-\nu_-+i\eta_+ & -60-\xi_-
\end{pmatrix},\\
R_2^{(-)}&=\frac{3}{34\sqrt{82}}
\begin{pmatrix}
60+\xi_+ & \eta_++i\nu_- \\
\eta_+-i\nu_- & -60-\xi_-
\end{pmatrix},
\end{split}\qquad
\begin{split}
R_3^{(-)}&=\frac{1}{17\sqrt{82}}
\begin{pmatrix}
-135+2\xi_+ & 2(1-i)\eta_+ \\
2(1+i)\eta_+ & 135-2\xi_-
\end{pmatrix},\\
R_4^{(-)}&=\frac{1}{2\sqrt{41}}
\begin{pmatrix}
4(\sqrt{41}-3) & -9(1-i) \\
-9(1+i) & 4(\sqrt{41}+3)
\end{pmatrix},
\end{split}
\end{equation}
\item $ $
\begin{equation}
\begin{split}
R_1^{\prime(+)}&=\frac{3}{34\sqrt{82}}
\begin{pmatrix}
-60+\xi_+ & \nu_++i\eta_- \\
\nu_+-i\eta_- & 60-\xi_-
\end{pmatrix},\\
R_2^{\prime(+)}&=\frac{3}{34\sqrt{82}}
\begin{pmatrix}
-60+\xi_+ & -\eta_--i\nu_+ \\
-\eta_-+i\nu_+ & 60-\xi_-
\end{pmatrix},
\end{split}\qquad
\begin{split}
R_3^{\prime(+)}&=\frac{1}{17\sqrt{82}}
\begin{pmatrix}
135+2\xi_+ & -2(1-i)\eta_- \\
-2(1+i)\eta_- & -135-2\xi_-
\end{pmatrix},\\
R_4^{\prime(+)}&=\frac{1}{2\sqrt{41}}
\begin{pmatrix}
4(\sqrt{41}-3) & -9(1-i) \\
-9(1+i) & 4(\sqrt{41}+3)
\end{pmatrix},
\end{split}
\end{equation}
\item $ $
\begin{equation}
\begin{split}
R_1^{\prime(-)}&=\frac{3}{34\sqrt{82}}
\begin{pmatrix}
60-\xi_- & -\nu_+-i\eta_- \\
-\nu_++i\eta_- & -60+\xi_+
\end{pmatrix},\\
R_2^{\prime(-)}&=\frac{3}{34\sqrt{82}}
\begin{pmatrix}
60-\xi_- & \eta_-+i\nu_+ \\
\eta_--i\nu_+ & -60+\xi_+
\end{pmatrix},
\end{split}\qquad
\begin{split}
R_3^{\prime(-)}&=\frac{1}{17\sqrt{82}}
\begin{pmatrix}
-135-2\xi_- & 2(1-i)\eta_- \\
2(1+i)\eta_- & 135+2\xi_+
\end{pmatrix},\\
R_4^{\prime(-)}&=\frac{1}{2\sqrt{41}}
\begin{pmatrix}
4(\sqrt{41}+3) & 9(1-i) \\
9(1+i) & 4(\sqrt{41}-3)
\end{pmatrix},
\end{split}
\end{equation}
\end{enumerate}
where $\eta_\pm=9(5\pm 2\sqrt{2})$, $\xi_\pm=\sqrt{2}(24\pm 17\sqrt{41})$, and $\nu_\pm=125\pm 18\sqrt{2}$. The above operators are calculated for the maximal admissible value of $c=225/41$, which is the same for all four GEWFs.
\end{Example}

Now, consider the opposite scenario: we are given a GEWF and asked what Hermitian orthonormal basis it corresponds to. The results of Proposition \ref{4Hk} show us how to recover $G_k$ from the semi-positive operators $R_k$.

\begin{Proposition}\label{4Gk}
An informationally complete GEWF $\mathcal{R}=\{R_k;\,k=1,\ldots,M=d^2\}$ produces four Hermitian orthonormal operator bases: $\mathcal{G}_\pm=\{\mathbb{I}_d/\sqrt{d},\pm G_k:\,k=1,\ldots,d^2-1\}$ and $\mathcal{G}^\prime_\pm=\{\mathbb{I}_d/\sqrt{d},\pm G_k^\prime:\,k=1,\ldots,d^2-1\}$ with
\begin{equation}
\begin{split}&
(i)\quad
G_k=\frac{1}{\sqrt{c}}\left[R_k+\frac{w_k}{\sqrt{d\gamma}-w_M}\left( R_M-\sqrt{\frac{\gamma}{d}}\mathbb{I}_d\right)\right],\\&
(ii)\quad
G_k^\prime=\frac{1}{\sqrt{c}}\left[R_k-\frac{w_k}{\sqrt{d\gamma}+w_M}\left(
R_M+\sqrt{\frac{\gamma}{d}}\mathbb{I}_d\right)\right].
\end{split}
\end{equation}
\end{Proposition}

\section{Mutually unbiased weighted frames}

Let us generalize the results from the previous section. Consider $N$ generalized equiangular weighted frames $\mathcal{R}_\alpha=\{R_{\alpha,k}:\,k=1,\ldots,M_\alpha\}$ on $\mathbb{C}^d$. Additionally, assume that the operators from two different $\mathcal{R}_\alpha$ are mutually unbiased \cite{Fickus,Goyeneche2} (or complementary \cite{PetzRuppert}); that is, $\Tr(R_{\alpha,k}R_{\beta,\ell})=f\Tr(R_{\alpha,k})\Tr(R_{\beta,\ell})$ for any $\beta\neq\alpha$. This results in the following definition.

\begin{Definition}\label{Def2}
The mutually unbiased equiangular weighted frames (MU GEWFs) are a collection $\mathcal{R}=\cup_{\alpha=1}^N\mathcal{R}_\alpha$ of $N$ GEWFs $\mathcal{R}_\alpha=\{R_{\alpha,k}:\,k=1,\ldots,M_\alpha\}$ such that
\begin{equation}\label{def2}
\begin{split}
\Tr(R_{\alpha,k})&=w_{\alpha,k},\\
\Tr(R_{\alpha,k}^2)&=c_\alpha+\frac{\gamma_\alpha-c_\alpha}{d\gamma_\alpha}w_{\alpha,k}^2,\\
\Tr(R_{\alpha,k}R_{\alpha,\ell})&=\frac{\gamma_\alpha-c_\alpha}{d\gamma_\alpha}w_{\alpha,k} w_{\alpha,\ell},\quad \ell\neq k,\\
\Tr(R_{\alpha,k}R_{\beta,\ell})&=\frac 1d w_{\alpha,k} w_{\beta,\ell},\quad \beta\neq \alpha,\\
\sum_{k=1}^{M_\alpha}w_{\alpha,k}R_{\alpha,k}&=\gamma_\alpha\mathbb{I}_d,\quad \gamma_\alpha>0.
\end{split}
\end{equation}
\end{Definition}

It is important to note that $f=1/d$ is the only solution that guarantees mutual unbiasedness of $\mathcal{R}_\alpha$, $\alpha=1,\ldots,N$. Analogically to eqs. (\ref{war1}), (\ref{war2}), and (\ref{c2}), we have
\begin{equation}\label{cgamma}
0<c_\alpha\leq\gamma_\alpha\min\left\{\frac{d-1}{M_\alpha-1},1\right\},\qquad \gamma_\alpha=\frac 1d \sum_{k=1}^{M_\alpha} w_{\alpha,k}^2,
\end{equation}
\begin{equation}\label{CS}
w_{\alpha,k}\geq\sqrt{\frac{dc_\alpha\gamma_\alpha}{(d-1)\gamma_\alpha+c_\alpha}}.
\end{equation}
For $N=1$, one recovers the GEWF from the previous section. 

Now, observe that, in each $\mathcal{R}_\alpha$, there is always one operator that is linearly dependent on $\mathbb{I}_d$ and the remaining elements of $\mathcal{R}_\alpha$. Therefore, the total number of operators in $\mathcal{R}=\cup_{\alpha=1}^N\mathcal{R}_\alpha$ is bounded by $|\mathcal{R}|=\sum_{\alpha=1}^{N}M_\alpha\leq d^2-1-N$.

\begin{Proposition}\label{IC}
The set $\widetilde{\mathcal{R}}=\{\mathbb{I}_d,\,R_{\alpha,k}:\,k=1,\ldots,M_\alpha-1;\,\alpha=1,\ldots,N\}$ with $|\widetilde{\mathcal{R}}|=d^2$ is an operator basis on $\mathbb{C}^d$.
\end{Proposition}

\begin{proof}
Following the method from Proposition 1, we define
\begin{equation}
R=r_0\mathbb{I}_d+\sum_{\alpha=1}^N\sum_{k=1}^{M_\alpha-1}r_{\alpha,k}R_{\alpha,k},\qquad r_0,r_{\alpha,k}\in\mathbb{R}.
\end{equation}
Now, if $R=0$, then
\begin{align}
\Tr(R)&=dr_0+\sum_{\alpha=1}^N\sum_{k=1}^{M_\alpha-1}r_{\alpha,k}w_{\alpha,k}=0,\label{rr1}\\
\Tr(RR_{\beta,\ell})&=c_\beta r_{\beta,\ell} -\frac{c_\beta w_{\beta,\ell}}{d\gamma_\beta}
\sum_{k=1}^{M_\alpha-1}r_{\beta,k}w_{\beta,k}=0.\label{rr2}
\end{align}
Observe that multiplying eq. (\ref{rr2}) by $w_{\beta,\ell}$ and taking the sum over $\ell$ gives
\begin{equation}
\sum_{\ell=1}^{M_\beta-1}w_{\beta,\ell}\Tr(RR_{\beta,\ell})
=\frac{c_\beta}{d}\left(d-1+\frac{w_{\beta,M_\beta}^2}{\gamma_\beta}\right)
\sum_{\ell=1}^{M_\beta-1}r_{\beta,\ell}w_{\beta,\ell}=0.
\end{equation}
The parameters $c_\beta>0$ and $w_{\beta,M_\beta}^2\leq\gamma_\beta$, and therefore this means that $\sum_{\ell=1}^{M_\beta-1}r_{\beta,\ell}w_{\beta,\ell}=0$. After applying this result, eqs. (\ref{rr1}) and (\ref{rr2}) are equivalent to $r_0=r_{\beta,\ell}=0$.
Finally, as $R=0$ implies that all $r_0=r_{\alpha,k}=0$,
the elements of $\widetilde{\mathcal{R}}$ are linearly independent. Moreover, $|\widetilde{\mathcal{R}}|=d^2$, and hence it is an operator basis.
\end{proof}

The construction of $R_{\alpha,k}$ follows directly from Proposition \ref{4Hk}. Analogically, we take a subset $\mathcal{G}=\{G_{\alpha,k}:\,k=1,\ldots,M_\alpha-1;\,\alpha=1,\ldots,N\}$ of traceless operators $G_{\alpha,k}$ from an orthonormal Hermitian operator basis that also contains $G_0=\mathbb{I}_d/\sqrt{d}$. We use $G_{\alpha,k}$ to introduce two families of traceless operators:
\begin{equation}\label{4Hak}
\left\{
\begin{split}
H_{\alpha,k}&=\frac{\sqrt{d\gamma_\alpha}}{w_{\alpha,k}}G_{\alpha,k}-\sum_{\ell=1}^{M_\alpha-1}
\frac{w_{\alpha,\ell}}{\sqrt{d\gamma_\alpha}-w_{\alpha,M_\alpha}} G_{\alpha,\ell},\\
H_{\alpha,M_\alpha}&=\sum_{\ell=1}^{M_\alpha-1}\frac{w_{\alpha,\ell}}{w_{\alpha,M_\alpha}} G_{\alpha,\ell},
\end{split}
\right.\qquad
\left\{
\begin{split}
H_{\alpha,k}^\prime&=\frac{\sqrt{d\gamma_\alpha}}{w_{\alpha,k}}G_{\alpha,k}-\sum_{\ell=1}^{M_\alpha-1} \frac{w_{\alpha,\ell}}{\sqrt{d\gamma_\alpha}+w_{\alpha,M_\alpha}} G_{\alpha,\ell},\\
H_{\alpha,M_\alpha}^\prime&=-\sum_{\ell=1}^{M_\alpha-1}\frac{w_{\alpha,\ell}}{w_{\alpha,M_\alpha}} G_{\alpha,\ell},
\end{split}
\right.
\end{equation}
where $k=1,\ldots,M_\alpha-1$. Therefore, for every $\alpha=1,\ldots,N$, one gets four GEWFs: $\mathcal{R}_\alpha^{(\pm)}=\{R_{\alpha,k}^{(\pm)}:\,k=1,\ldots,M_\alpha\}$ and $\mathcal{R}_\alpha^{\prime(\pm)}=\{R_{\alpha,k}^{\prime(\pm)}:\,k=1,\ldots,M_\alpha\}$. The elements of these sets are given by
\begin{equation}\label{gen2}
R_{\alpha,k}^{(\pm)}=\frac{w_{\alpha,k}}{d}\left(\mathbb{I}_d
\pm\sqrt{\frac{dc_\alpha}{\gamma_\alpha}}H_{\alpha,k}\right),
\end{equation}
\begin{equation}\label{gen2A}
R_{\alpha,k}^{\prime(\pm)}=\frac{w_{\alpha,k}}{d}\left(\mathbb{I}_d
\pm\sqrt{\frac{dc_\alpha}{\gamma_\alpha}}H_{\alpha,k}^\prime\right).
\end{equation}
In analogy to eq. (\ref{cc}), one has
\begin{equation}\label{cc2}
c_\alpha^{(+)}=\frac{\gamma_\alpha}{d\lambda_{\alpha,\min}^2},\qquad c_\alpha^{(-)}=\frac{\gamma_\alpha}{d\lambda_{\alpha,\max}^2},
\end{equation}
where $\lambda_{\alpha,\min}$ and $\lambda_{\alpha,\max}$ are the minimal and maximal eigenvalues from among all eigenvalues of $H_{\alpha,k}$ or $H_{\alpha,k}^\prime$, $k=1,\ldots,M_\alpha$.
Therefore, $R_{\alpha,k}^{(+)}$ and $R_{\alpha,k}^{\prime(+)}$ are positive semidefinite if and only if $0<c_\alpha\leq c_\alpha^{(+)}$. In full analogy, $R_{\alpha,k}^{(-)}$ and $R_{\alpha,k}^{\prime(-)}$ are positive semidefinite if and only if $0<c_\alpha\leq c_\alpha^{(-)}$.

Obviously, every collection $\mathcal{R}_\pm=\cup_{\alpha=1}^N\mathcal{R}_\alpha^{(\pm)}$ and $\mathcal{R}_\pm^\prime=\cup_{\alpha=1}^N\mathcal{R}_\alpha^{\prime(\pm)}$ form four sets of mutually unbiased (MU) GEWFs. However, these are not the only possible choices. In fact, every collection $\mathcal{R}=\cup_{\alpha=1}^N\mathcal{R}_\alpha$ with $\mathcal{R}_\alpha\in\{\mathcal{R}_\alpha^{(\pm)},\mathcal{R}_\alpha^{\prime(\pm)}\}$ defines a MU GEWF due to
\begin{equation}
\Tr(H_{\alpha,k}H_{\beta,\ell})=\Tr(H_{\alpha,k}^\prime H_{\beta,\ell}^\prime)=
\Tr(H_{\alpha,k}H_{\beta,\ell}^\prime)=0
\end{equation}
for any $\alpha\neq\beta$. This means that a single choice of $\mathcal{G}$ allows us to construct $4N$ sets $\mathcal{R}$.

Conversely, a single informationally (over)complete set $\mathcal{R}$ of MU GEWFs with elements $R_{\alpha,k}$ leads to $4N$ orthonormal Hermitian operator bases $\mathcal{G}=\cup_{\alpha=0}^N\mathcal{G}_\alpha$ with $\mathcal{G}_0=\{\mathbb{I}_d/\sqrt{d}\}$ and $\mathcal{G}_\alpha\in\{\mathcal{G}_\alpha^{(\pm)},\mathcal{G}_\alpha^{\prime(\pm)}\}$. From Proposition \ref{4Gk}, these sets of linearly independent operators read as follows:
\begin{enumerate}[label=(\it\roman*)]
\item $\mathcal{G}_\alpha^{(\pm)}=\{\pm G_{\alpha,k}:\,k=1,\ldots,M_\alpha\}$ with
\begin{equation}\label{Gak}
G_{\alpha,k}=\frac{1}{\sqrt{c_\alpha}}\left[R_{\alpha,k}+\frac{w_{\alpha,k}}
{\sqrt{d\gamma_\alpha}- w_{\alpha,M_\alpha}}\left(R_{\alpha,M_\alpha}
-\sqrt{\frac{\gamma_\alpha}{d}}\mathbb{I}_d\right)\right],
\end{equation}
\item $\mathcal{G}_\alpha^{\prime(\pm)}=\{\pm G_{\alpha,k}^\prime:\,k=1,\ldots,M_\alpha\}$, where
\begin{equation}
G_{\alpha,k}^\prime=\frac{1}{\sqrt{c_\alpha}}\left[R_{\alpha,k}-\frac{w_{\alpha,k}}
{\sqrt{d\gamma_\alpha}+w_{\alpha,M_\alpha}}\left(R_{\alpha,M_\alpha}
+\sqrt{\frac{\gamma_\alpha}{d}}\mathbb{I}_d\right)\right].
\end{equation}
\end{enumerate}

\begin{Example}\label{Exak}
Consider the qubit scenario ($d=2$), where we have $N=2$ sets of operators $R_{\alpha,k}=R_{\alpha,k}^{(+)}$ of the size $M_1=2$ and $M_2=3$. The traces of $R_{\alpha,k}$ are fixed, so that
\[w_{1,1}=1,\qquad w_{1,2}=w_{2,1}=\sqrt{3},\qquad w_{2,2}=2,\qquad w_{2,3}=3,\]
and hence $\gamma_1=2$, $\gamma_2=8$. Let us chooce the parameters $c_1=c_2=2/3$, which guarantees the semi-positivity of all the elements $R_{\alpha,k}$. This results in two MU GEWFs:
\begin{equation}
R_{1,1}=\frac 12 \begin{pmatrix}
1 & -1 \\ -1 & 1
\end{pmatrix},\qquad 
R_{1,2}=\frac{1}{2\sqrt{3}} \begin{pmatrix}
3 & 1 \\ 1 & 3
\end{pmatrix},
\end{equation}
\begin{equation}
R_{2,1}=\frac{1}{4\sqrt{3}} \begin{pmatrix}
6-2\sqrt{3} & -i \\ i & 6+2\sqrt{3}
\end{pmatrix},\qquad 
R_{2,2}=\frac{1}{2} \begin{pmatrix}
2 & i \\ -i & 2
\end{pmatrix},
\qquad 
R_{2,3}=\frac{1}{4\sqrt{3}} \begin{pmatrix}
2+6\sqrt{3} & -i\sqrt{3} \\ i\sqrt{3} & -2+6\sqrt{3}
\end{pmatrix}.
\end{equation}
\end{Example}


Now, let us further analyze the relations between $R_{\alpha,k}^{(+)}$ and $R_{\alpha,k}^{(-)}$ under the assumption that $c_\alpha\leq\min\{c_\alpha^{(\pm)}\}$, so that both $R_{\alpha,k}^{(\pm)}\geq 0$. Then,
\begin{equation}
\Tr(R_{\alpha,k}^{(+)}R_{\alpha,k}^{(-)})=-c_\alpha+\frac{\gamma_\alpha+c_\alpha}{d\gamma_\alpha}
w_{\alpha,k}^2,
\end{equation}
\begin{equation}
\Tr(R_{\alpha,k}^{(+)}R_{\alpha,\ell}^{(-)})=\frac{\gamma_\alpha+c_\alpha}{d\gamma_\alpha}
w_{\alpha,k}w_{\alpha,\ell},\qquad k\neq\ell.
\end{equation}
Note that the above results differ from $\Tr(R_{\alpha,k}^2)$ and $\Tr(R_{\alpha,k}R_{\alpha,\ell})$, respectively, in the sign before $c_\alpha$. Moreover, $R_{\alpha,k}^{(+)}$ and $R_{\alpha,k}^{(-)}$ are orthogonal in the Hilbert-Schmidt inner product provided that
\begin{equation}
w_{\alpha,k}=\sqrt{\frac{d\gamma_\alpha}{M_\alpha}},\qquad c_\alpha=\frac{\gamma_\alpha}{M_\alpha-1}.
\end{equation}
For this choice of parameters,
\begin{equation}
\Tr(R_{\alpha,k}^{(+)2})=\frac{2\gamma_\alpha}{M_\alpha},\qquad 
\Tr(R_{\alpha,k}^{(+)}R_{\alpha,\ell}^{(+)})=\frac{M_\alpha-2}{M_\alpha}\gamma_\alpha,\qquad k\neq\ell,
\end{equation}
which means that $R_{\alpha,k}^{(+)}$ are mutually orthogonal for $M_\alpha=2$.


\section{Relation to equiangular tight frames}

When constructing equiangular tight frames, one first considers collections of equiangular lines $\mathcal{E}_\alpha=\{E_{\alpha,k}:\,k=1,\ldots,M_\alpha\}$, where $\Tr(E_{\alpha,k}E_{\alpha,\ell})=y_\alpha\Tr(E_{\alpha,k})\Tr(E_{\alpha,\ell})$ for $k\neq\ell$. If additionally $\sum_{k=1}^{M_\alpha}E_{\alpha,k}=\kappa_\alpha\mathbb{I}_d$ for some positive constant $\kappa_\alpha$, then $\mathcal{E}_\alpha$ forms an equiangular tight frame \cite{Strohmer}. The difference between tight and weighted frames is in the admissible values of weights, which here have to be only $\alpha$-dependent. This results in every operator $E_{\alpha,k}\in\mathcal{E}_\alpha$ to be of equal trace. That property holds also for their generalizations.

\begin{Definition}(\cite{GEAM,conical})
The mutually unbiased equiangular tight frames (MU GETFs) are a collection $\mathcal{E}=\cup_{\alpha=1}^N\mathcal{E}_\alpha$ of $N$ GETFs $\mathcal{E}_\alpha=\{E_{\alpha,k}:\,k=1,\ldots,M_\alpha\}$ such that
\begin{equation}\label{def4}
\begin{split}
\Tr(E_{\alpha,k})&=a_\alpha,\\
\Tr(E_{\alpha,k}^2)&=a_\alpha^2b_\alpha,\\
\Tr(E_{\alpha,k}E_{\alpha,\ell})&=a_\alpha^2\frac{M_\alpha-db_\alpha}{d(M_\alpha-1)},\quad \ell\neq k,\\
\Tr(E_{\alpha,k}E_{\beta,\ell})&=\frac 1d a_\alpha a_\beta,\quad \beta\neq \alpha,\\
\sum_{k=1}^{M_\alpha}E_{\alpha,k}&=\kappa_\alpha\mathbb{I}_d,\quad\kappa_\alpha>0,
\end{split}
\end{equation}
where the free parameters $b_\alpha$ belong to the range
\begin{equation}\label{brange2}
\frac{1}{d}<b_\alpha\leq\frac 1d \min\{d,M_\alpha\}.
\end{equation}
\end{Definition}

The MU GETFs are a relatively general concept. They include many popular operator families as special cases. Some examples are symmetric, informationally complete (SIC) POVMs \cite{Renes}, general SIC POVMs \cite{Gour}, $(N,M)$-POVMs \cite{SIC-MUB}, generalized symmetric measurements \cite{SIC-MUB_general}, generalized equiangular measurements \cite{GEAM}, as well as mutually unbiased bases, measurements \cite{Kalev}, and equiangular tight frames \cite{Fickus}.

From the ways they are defined, it is obvious that MU GETFs are a special case of MU GEWFs. Namely, the frames in MU GEWFs are tight if $w_{\alpha,k}\equiv w_\alpha$. Then, the relations between their characteristic parameters $(w_\alpha,c_\alpha,\gamma_\alpha)\leftrightarrow(a_\alpha,b_\alpha,\kappa_\alpha)$ are as follows,
\begin{equation}
a_\alpha=w_\alpha,\qquad b_\alpha=\frac{\gamma_\alpha-c_\alpha} {d\gamma_\alpha}+\frac{c_\alpha}{w_\alpha^2},\qquad
\kappa_\alpha=\frac{w_\alpha M_\alpha}{d},
\end{equation}
or
\begin{equation}
w_\alpha=a_\alpha,\qquad c_\alpha=a_\alpha^2M_\alpha\frac{db_\alpha-1} {d(M_\alpha-1)},\qquad\gamma_\alpha=\frac{a_\alpha^2M_\alpha}{d}.
\end{equation}

\section{Relation to conical 2-designs}

Conical designs generalize complex projective designs to operators of arbitary rank from the cone of mixed quantum states $\{\lambda\rho:\,\lambda>0,\,\rho\,-{\rm\,density\,operator}\}$. They are defined as follows.

\begin{Definition}(\cite{Graydon})
The conical 2-design is a collection $\mathcal{E}=\{E_k;\,k=1,\ldots,M\}$ of $M\geq d^2$ positive semi-definite operators $E_k$ that satisfy
\begin{equation}\label{conic}
\sum_{k=1}^ME_k\otimes E_k=\kappa_+\mathbb{I}_d\otimes\mathbb{I}_d+\kappa_-\mathbb{F}_d
\end{equation}
with real parameters $\kappa_+\geq\kappa_->0$ and the flip operator $\mathbb{F}_d=\sum_{m,n=0}^{d-1}|m\>\<n|\otimes|n\>\<m|$.
\end{Definition}

From this definition, it follows that a conical 2-design is a family $\mathcal{E}$ of positive semi-definite operators $E_k$ such that $E_k$ span the operator space on $\mathbb{C}^d$ and the sum $\sum_{k=1}^ME_k\otimes E_k$ commutes with $U\otimes U$, where $U$ is an arbitrary unitary operator on $\mathbb{C}^d$. In addition, eq. (\ref{conic}) can be equivalently rewritten in terms of linear mappings. Indeed, if $\Phi[X]=\sum_{k=1}^ME_k\Tr(E_kX)$, then $E_k$ form a conical 2-design if and only if \cite{SIC-MUB_general}
\begin{equation}\label{kap}
\Phi=\kappa_+d\Phi_0+\kappa_-\oper,\qquad \kappa_+\geq\kappa_->0.
\end{equation}
In other words, the map $\Phi$ has to be a linear combination of the identity map $\oper$ and the maximally depolarizing channel $\Phi_0[X]=\mathbb{I}_d\Tr(X)/d$. Some examples of conical 2-designs are homogeneous designs (with constant $\Tr(E_k)$ and $\Tr(E_k^2)$ \cite{Graydon}), complex projective designs \cite{Neumaier,Hoggar,Scott}, and equidistant MU GETFs \cite{GEAM,conical}. Here, we show that MU GETFs are another special class of conical 2-designs, allowing for operators of unequal traces.

\begin{Proposition}
The MU GEWFs $R_{\alpha,k}$ introduced in Definition \ref{Def2} form a conical 2-design if $|\mathcal{R}|=d^2-1-N$ and $c_\alpha\equiv c$ for all $\alpha=1,\ldots,N$. Then,
\begin{equation}\label{kappas}
\kappa_-=c,\qquad \kappa_+=\frac{\Gamma-c}{d},
\end{equation}
where $\Gamma=\sum_{\alpha=1}^N\gamma_\alpha$.
\end{Proposition}

\begin{proof}
From Proposition \ref{IC}, we see that $\mathcal{R}$ spans the operator space on $\mathbb{C}^d$ if and only if $|\mathcal{R}|=d^2-1-N$.
Now, following the methods from ref. \cite{conical}, we first use $\mathcal{R}_\alpha$ to introduce the corresponding linear maps $\Phi_\alpha[X]=\sum_{k=1}^{M_\alpha}R_{\alpha,k}\Tr(R_{\alpha,k}X)$. Then, we calculate
\begin{equation}
\begin{split}
\Tr(\Phi_\alpha[X]R_{\alpha,\ell})=c_\alpha\Tr(R_{\alpha,\ell} X)
+\frac{\gamma_\alpha-c_\alpha}{d}w_{\alpha,\ell}\Tr(X).
\end{split}
\end{equation}
and, for any $\beta\neq\alpha$,
\begin{equation}
\begin{split}
\Tr(\Phi_\beta[X]R_{\alpha,\ell})&=\frac{\gamma_\beta}{d}w_{\alpha,\ell}\Tr(X).
\end{split}
\end{equation}
Hence, the total map $\Phi=\sum_{\alpha=1}^N\Phi_\alpha$ behaves as follows,
\begin{equation}\label{phi}
\Tr(\Phi[X]P_{\alpha,k})=c_\alpha\Tr(R_{\alpha,k}X)+\frac{w_{\alpha,\ell}}{d}
(\Gamma-c_\alpha)\Tr(X),
\end{equation}
where $\Gamma=\sum_{\beta=1}^N\gamma_\beta$. Equivalently, one has
\begin{equation}
\Tr\left\{\left[c_\alpha\oper+(\Gamma-c_\alpha)\Phi_0-\Phi\right][X]R_{\alpha,k}\right\}=0,
\end{equation}
This trace relation has to be satisfied for all $X$ and $R_{\alpha,k}$. Therefore, the total map sums up to
\begin{equation}
\Phi=c_\alpha\oper+(\Gamma-c_\alpha)\Phi_0.
\end{equation}
Now, if we compare this $\Phi$ with the map $\Phi=\kappa_+d\Phi_0+\kappa_-\oper$ from eq. (\ref{kap}), we see that $R_{\alpha,k}$ form a conical 2-design if and only if $c_\alpha\equiv c$ for all $\alpha=1,\ldots,N$. The characterizing constants $\kappa_\pm$ are then
\begin{equation}
\kappa_+=\frac{\Gamma-c}{d},\qquad \kappa_-=c.
\end{equation}
Obviously, $\kappa_->0$. Moreover, from eq. (\ref{CS}), it follows that
\begin{equation}
w_{\alpha,k}^2\geq \frac{d\kappa_-\gamma_\alpha}{(d-1)\gamma_\alpha+\kappa_-}.
\end{equation}
Taking the sum over $k$ results in
\begin{equation}
d\gamma_\alpha\geq\frac{d\kappa_-\gamma_\alpha M_\alpha}
{(d-1)\gamma_\alpha+\kappa_-},
\end{equation}
which can be rewritten into
\begin{equation}
(M_\alpha-1)\kappa_-\leq(d-1)\gamma_\alpha.
\end{equation}
Finally, taking the sum over $\alpha$ and recalling that $|\mathcal{R}|=\sum_{\alpha=1}^NM_\alpha=d^2+N-1$, we recover $\kappa_-\leq\kappa_+$.
\end{proof}

In ref. \cite{conical}, the inverse implication has been shown. Therefore, we have proven the following equivalence theorem.

\begin{Theorem}
Consider $N$ sets $\mathcal{R}_\alpha=\{R_{\alpha,k};\,k=1,\ldots,M_\alpha\}$ of linearly independent semi-positive operators such that $\sum_{k=1}^{M_\alpha}\Tr(R_{\alpha,k})R_{\alpha,k}=\gamma_\alpha\mathbb{I}_d$ for some $\gamma_\alpha>0$ and together span the operator space on $\mathbb{C}^d$. The collection $\mathcal{R}=\cup_{\alpha=1}^N\mathcal{R}_\alpha$ with $|\mathcal{R}|=d^2+N-1$ is a conical 2-design if and only if it is a maximal set of MU GEWFs characterized by $c_\alpha\equiv c$.
\end{Theorem}

Interestingly, the condition that $c_\alpha\equiv c$ for all $\alpha=1,\ldots,N$ is equivalent to uniform distance scaling between renormalized operators $R_{\alpha,k}$ from the same GEWF. To observe this, let us introduce $\widetilde{R}_{\alpha,k}=R_{\alpha,k}/w_{\alpha,k}$, so that $\Tr(\widetilde{R}_{\alpha,k})=1$. For a fixed $\alpha$, the Frobenius distance between $\widetilde{R}_{\alpha,k}$ is equal to
\begin{equation}
D_2^2(\widetilde{R}_{\alpha,k},\widetilde{R}_{\alpha,\ell})=\frac 12\|\widetilde{R}_{\alpha,k}-\widetilde{R}_{\alpha,\ell}\|_2^2
=\frac 12\Tr[(\widetilde{R}_{\alpha,k}-\widetilde{R}_{\alpha,\ell})^2]=
c_\alpha h_{\alpha;k,\ell},
\end{equation}
where $\|X\|_2=\sqrt{\Tr[X^\dagger X]}$ is the Frobenius norm of $X$ and
\begin{equation}
h_{\alpha;k,\ell}=\frac 12 \left(\frac{1}{w_{\alpha,k}^2}+\frac{1}{w_{\alpha,\ell}^2}\right)
\end{equation}
is a function of two weights. Now, if this distance depends on $h_{\alpha;k,\ell}$ rescaled by the same constant $c_\alpha\equiv c$, then the MU GEWFs form a conical 2-design. Notably, for $w_{\alpha,k}\equiv w_\alpha$, this result reproduces the equidistance property for MU GETFs \cite{GEAM_coherence}.

\begin{Example}
For the GEWFs from Example \ref{excon} ($i$--$iv$),
\begin{equation}
\sum_{k=1}^4R_k\otimes R_k=\frac{25}{41}
\begin{pmatrix}
25 & 0 & 0 & 0 \\
0 & 16 & 9 & 0 \\
0 & 9 & 16 & 0 \\
0 & 0 & 0 & 25
\end{pmatrix}.
\end{equation}
Hence, they are conical 2-designs with $\kappa_+=(\gamma-c)/d=400/41$ and $\kappa_-=c=225/41$.
\end{Example}

To illustrate what happens when the operators $R_k$ do not form an informationally complete set, we present the following example.

\begin{Example}
Take $d=2$, $M=3$, and $\gamma=2$. For the orthonormal traceless operators, we choose the normalized Pauli matrices $G_1=\sigma_1/\sqrt{2}$, $G_2=\sigma_3/\sqrt{2}$, so that they are not only Hermitian but also have real entries. The traces of $R_k$ are fixed in such a way that $w_1=w_3=1$ and $w_2=\sqrt{2}$. Then, using $H_k$ from Proposition 2 (i), we find
\begin{equation}
\begin{split}
R_1&=\frac{1}{2\sqrt{2}}
\begin{pmatrix}
\sqrt{2}(1-\sqrt{c}) & \sqrt{c} \\
\sqrt{c} & \sqrt{2}(1+\sqrt{c})
\end{pmatrix},\\
R_2&=\frac 12
\begin{pmatrix}
\sqrt{2} & -\sqrt{c} \\
-\sqrt{c} & \sqrt{2}
\end{pmatrix},\\
R_3&=\frac{1}{2\sqrt{2}}
\begin{pmatrix}
\sqrt{2}(1+\sqrt{c}) & \sqrt{c} \\
\sqrt{c} & \sqrt{2}(1-\sqrt{c})
\end{pmatrix}.
\end{split}
\end{equation}
These operators are semi-positive if and only if $0<c\leq 2/3$. Note that
\begin{equation}
\frac 14 \sum_{k=1}^3R_k\otimes R_k=\frac 18 
\begin{pmatrix}
2+c & 0 & 0 & c \\
0 & 2-c & c & 0 \\
0 & c & 2-c & 0 \\
c & 0 & 0 & 2+c
\end{pmatrix}
\end{equation}
is an X-shaped quantum state (X-state) \cite{Eberly2}.
\end{Example}

\begin{Example}
Take the MU GEWFs from Example \ref{Exak}.
Observe that
\begin{equation}
\phi_1=R_{1,1}\otimes R_{1,1}+R_{1,2}\otimes R_{1,2}=
\frac 13 \begin{pmatrix}
3 & 0 & 0 & 1 \\
0 & 3 & 1 & 0 \\
0 & 1 & 3 & 0 \\
1 & 0 & 0 & 3
\end{pmatrix}
\end{equation}
and
\begin{equation}
\phi_2=R_{1,1}\otimes R_{1,1}+R_{1,2}\otimes R_{1,2}=
\frac 13 \begin{pmatrix}
13 & 0 & 0 & -1 \\
0 & 11 & 1 & 0 \\
0 & 1 & 11 & 0 \\
-1 & 0 & 0 & 13
\end{pmatrix}
\end{equation}
are both X-shaped. Hence, neither $\mathcal{R}_1$ nor $\mathcal{R}_2$ is a conical 2-design. However, $\mathcal{R}=\mathcal{R}_1\cup\mathcal{R}_2$ is a conical 2-design with $\kappa_+=14/3$ and $\kappa_-=2/3$, which follows from the sum
\begin{equation}
\phi_1+\phi_2=\frac 13 \begin{pmatrix}
16 & 0 & 0 & 0 \\
0 & 14 & 2 & 0 \\
0 & 2 & 14 & 0 \\
0 & 0 & 0 & 16
\end{pmatrix}=\frac{14}{3}\mathbb{I}_2\otimes\mathbb{I}_2+\frac{2}{3}\mathbb{F}_2.
\end{equation}
\end{Example}

\section{Relation to quantum measurements}

In quantum information, the knowledge about quantum states is acquired through a measurement process. Measurement operators on qudits are represented by a positive, operator-valued measure (POVM), which is a collection $\mathcal{P}=\{P_k:\,k=1,\ldots,L\}$ of
positive semidefinite operators $P_k$ on $\mathbb{C}^d$ that are a resolution of the identity operator ($\sum_{k=1}^LP_k=\mathbb{I}_d$). Obviously, in general, MU GEWFs are not POVMs.

Consider a special case of MU GEWFs, where each GEWF consists of operators of equal traces. In other words, we require that $w_{\alpha,k}=w_\alpha$ for all $k=1,\ldots,M_\alpha$, so that $\sum_{k=1}^{M_\alpha}R_{\alpha,k}=\gamma_\alpha\mathbb{I}_d$ with $\gamma_\alpha=w_\alpha M_\alpha$. Then, all the operators $R_{\alpha,k}$ sum up to the identity operator provided that $\sum_{\alpha=1}^N\gamma_\alpha=1$. This way, we recover a generalized equiangular measurement (GEAM) \cite{GEAM}, whose applications and relations to conical 2-designs have been analyzed in refs. \cite{GEAM_coherence,GEAM_Pmaps,GEAM_kmaps,conical}.

Actually, any MU GEWFs define the POVM elements via
\begin{equation}\label{povm}
P_{\alpha,k}=w_{\alpha,k}R_{\alpha,k},
\end{equation}
where we impose the additional condition $\Gamma=\sum_{\alpha=1}^N\gamma_\alpha=1$ to avoid redefining $\gamma_\alpha$. However, a more general construction with $\Gamma\neq 1$ is straightforward, as it only requires an additional rescaling by $1/\Gamma$. This definition allows for unequal traces even among the operators in a single $\mathcal{P}_\alpha=\{P_{\alpha,k}:\,k=1,\ldots,M_\alpha\}$. Following eq. (\ref{def2}), we derive the trace properties with respect to $P_{\alpha,k}$;
\begin{equation}\label{def3}
\begin{split}
\Tr(P_{\alpha,k})&=a_{\alpha,k},\\
\Tr(P_{\alpha,k}^2)&=c_\alpha a_{\alpha,k}
+\frac{\gamma_\alpha-c_\alpha}{d\gamma_\alpha}a_{\alpha,k}^2,\\
\Tr(P_{\alpha,k}P_{\alpha,\ell})&=\frac{\gamma_\alpha-c_\alpha}{d\gamma_\alpha}a_{\alpha,k} a_{\alpha,\ell},\quad \ell\neq k,\\
\Tr(P_{\alpha,k}P_{\beta,\ell})&=\frac 1d a_{\alpha,k} a_{\beta,\ell},\quad \beta\neq \alpha,
\end{split}
\end{equation}
where $a_{\alpha,k}=w_{\alpha,k}^2$. Note that eqs. (\ref{cgamma}) and (\ref{CS}) translate into
\begin{equation}\label{ca}
0<c_\alpha\leq\gamma_\alpha\min\left\{\frac{d-1}{M_\alpha-1},1\right\},\qquad \gamma_\alpha=\frac 1d \sum_{k=1}^{M_\alpha}a_{\alpha,k},\qquad a_{\alpha,k}\geq \frac{dc_\alpha\gamma_\alpha}{(d-1)\gamma_\alpha+c_\alpha},
\end{equation}
and the elements of $\mathcal{P}_\alpha$ sum up to $\sum_{k=1}^{M_\alpha}P_{\alpha,k}=\gamma_\alpha\mathbb{I}_d$.

\begin{Remark}
If $c_\alpha\equiv c$ for all $\alpha=1,\ldots,N$, then the POVM $\mathcal{P}$ with elements $P_{\alpha,k}=w_{\alpha,k}R_{\alpha,k}$ is constructed from a conical 2-design. However, $\mathcal{P}$ is not a conical 2-design itself. Instead, its elements satisfy a similar relation,
\begin{equation}
\sum_{\alpha=1}^N\sum_{k=1}^{M_\alpha}\frac{1}{a_{\alpha,k}}P_{\alpha,k}\otimes P_{\alpha,k}=\kappa_+\mathbb{I}_d\otimes\mathbb{I}_d+\kappa_-\mathbb{F}_d
\end{equation}
with $\kappa_-=c$ and $\kappa_+=(1-c)/d$.
\end{Remark}

Now, if $|\mathcal{P}|=d^2+N-1$, when $P_{\alpha,k}$ form an informationally (over)complete set, and hence any quantum state $\rho$ can be expanded into
\begin{equation}\label{rho}
\rho=\sum_{\alpha=1}^N\sum_{k=1}^{M_\alpha}p_{\alpha,k}Q_{\alpha,k}.
\end{equation}
The nonnegative numbers $p_{\alpha,k}=\Tr(P_{\alpha,k}\rho)$ are probability outcomes, whereas $Q_{\alpha,k}$ form a frame dual to $P_{\alpha,k}$. For $N\geq 2$, the choice of a dual frame is ambiguous. Below, we present one possible choice.

\begin{Proposition}\label{PropQ}
The operators
\begin{equation}
Q_{\alpha,k}=\frac{1}{c_\alpha} \left(\frac{1}{a_{\alpha,k}}P_{\alpha,k}-\frac{1-c_\alpha}
{d}\mathbb{I}_d\right)
\end{equation}
form a dual frame to $P_{\alpha,k}$.
\end{Proposition}

\begin{proof}
Assume that the dual frame operators have the form
\begin{equation}
Q_{\alpha,k}=x_{\alpha,k}P_{\alpha,k}-y_\alpha\mathbb{I}_d
\end{equation}
and calculate the expectation value of $P_{\beta,\ell}$ in the state $\rho$ from eq. (\ref{rho}),
\begin{equation}
p_{\beta,\ell}=\Tr(P_{\beta,\ell}\rho)=\frac{a_{\beta,\ell}}{d}\left[\sum_{\alpha=1}^N\sum_{k=1}^{M_\alpha}
a_{\alpha,k}p_{\alpha,k}x_{\alpha,k}-\frac{c_\beta}{\gamma_\beta}
\sum_{k=1}^{M_\alpha}a_{\alpha,k}p_{\alpha,k}x_{\alpha,k}+dc_\beta p_{\beta,\ell}x_{\beta,\ell}
-d\sum_{\alpha=1}^N\gamma_\alpha y_\alpha\right].
\end{equation}
On the right hand-side, there is only one term that explicitly depends on $p_{\beta,\ell}$. If we take
\begin{equation}
x_{\beta,\ell}=\frac{1}{a_{\beta,\ell}c_\beta},
\end{equation}
then the terms $p_{\beta,\ell}$ cancel out, and hence we now have
\begin{equation}
\begin{split}
0&=\sum_{\alpha=1}^N\sum_{k=1}^{M_\alpha}
a_{\alpha,k}p_{\alpha,k}x_{\alpha,k}-\frac{c_\beta}{\gamma_\beta}
\sum_{k=1}^{M_\alpha}a_{\alpha,k}p_{\alpha,k}x_{\alpha,k}
-d\sum_{\alpha=1}^N\gamma_\alpha y_\alpha\\
&=\sum_{\alpha=1}^N\frac{\gamma_\alpha}{c_\alpha}-1-d\sum_{\alpha=1}^N\gamma_\alpha y_\alpha
=\sum_{\alpha=1}^N\gamma_\alpha\left(\frac{1}{c_\alpha}-1 -dy_\alpha\right).
\end{split}
\end{equation}
This sum vanishes if all its components are equal to zero for any $\gamma_\alpha$, from which we recover
\begin{equation}
y_\alpha=\frac{1-c_\alpha}{dc_\alpha}.
\end{equation}
\end{proof}

\section{Properties and applications}

\subsection{Index of coincidence}

Aside from the trace properties given in their definition, the MU GEWFs also satisfy a more refined property when traced with an arbitrary linear operator $X$ on $\mathbb{C}^d$.

\begin{Proposition}\label{XIOC}
If $\mathcal{R}$ is a conical 2-design, then
\begin{equation}\label{sumL}
\sum_{\beta=1}^L\sum_{\ell=1}^{M_\beta}|\Tr(R_{\beta,\ell}X)|^2
\leq\frac{\Gamma_L-c}{d}|\Tr(X)|^2-\kappa_-\Tr(X^\dagger X),
\end{equation}
where $\Gamma_L=\sum_{\alpha=1}^L\gamma_\alpha$.
The equality is reached for $L=N$, in which case
\begin{equation}\label{sumN}
\sum_{\beta=1}^N\sum_{\ell=1}^{M_\beta}|\Tr(R_{\beta,\ell}X)|^2
=\kappa_+|\Tr(X)|^2-\kappa_-\Tr(X^\dagger X).
\end{equation}
\end{Proposition}

\begin{proof}
Any operator $X$ on $\mathbb{C}^d$ can be expanded in the frame $Q_{\alpha,k}$ from Proposition \ref{PropQ}, so that
\begin{equation}
X=\sum_{\alpha=1}^N\sum_{k=1}^{M_\alpha}x_{\alpha,k}Q_{\alpha,k}
\end{equation}
with complex coefficients $x_{\alpha,k}$ that fully characterize $X$. Using the trace properties of $Q_{\alpha,k}$, namely $\Tr(Q_{\alpha,k})=1$ and
\begin{equation}\label{QQ}
\Tr(Q_{\alpha,k}Q_{\beta,\ell})=\frac 1d+\frac{1}{c_\alpha c_\beta}\left(
\frac{\Tr(P_{\alpha,k}P_{\beta,\ell})}{a_{\alpha,k}a_{\beta,\ell}}-\frac 1d\right)=\frac 1d +\frac{\delta_{\alpha\beta}}{c_\alpha}\left(
\frac{\delta_{k\ell}}{a_{\alpha,k}}-\frac{1}{d\gamma_\alpha}\right),
\end{equation}
we derive the formula
\begin{equation}
\Tr(XP_{\beta,\ell})=x_{\beta,\ell}+\frac{a_{\beta,\ell}}{d\gamma_\beta}
[\gamma_\beta \Tr(X)-x_\beta],
\end{equation}
where $x_\beta\equiv\sum_{k=1}^{M_\beta}x_{\beta,\ell}$ and $\Tr(X)=\sum_{\beta=1}^Nx_{\beta}$. Now, direct calculations lead to
\begin{equation}\label{lsum}
\sum_{\ell=1}^{M_\beta}|\Tr(XR_{\beta,\ell})|^2
=\sum_{\ell=1}^{M_\beta}\frac{1}{a_{\beta,\ell}}|\Tr(XP_{\beta,\ell})|^2
=\frac{\gamma_\beta}{d}|\Tr(X)|^2-\frac{|x_\beta|^2}{d\gamma_\beta}+\sum_{\ell=1}^{M_\beta}
\frac{|x_{\beta,\ell}|^2}{a_{\beta,\ell}},
\end{equation}
after recalling from eq. (\ref{povm}) that $P_{\alpha,k}=\sqrt{a_{\alpha,k}}R_{\alpha,k}$.
Next, to express the right hand-side in terms of $X$, we first note that
\begin{equation}\label{QP}
\Tr(Q_{\alpha,k}P_{\alpha,k})=\delta_{k\ell}-\frac{(1-\gamma_\alpha)a_{\alpha,k}}
{d\gamma_\alpha},\qquad
\Tr(Q_{\alpha,k}P_{\beta,\ell})=\frac{a_{\beta,\ell}}{d}\quad (\alpha\neq\beta),
\end{equation}
and then
\begin{equation}\label{XX}
\Tr(X^\dagger X)=\frac 1d \left[|\Tr(X)|^2-\sum_{\alpha=1}^N\frac{|x_\alpha|^2}{c_\alpha \gamma_\alpha}\right]
+\sum_{\alpha=1}^N\sum_{k=1}^{M_\alpha}\frac{|x_{\alpha,k}|^2}{a_{\alpha,k}c_\alpha}.
\end{equation}
Therefore, if $c_\alpha\equiv c$, then eq. (\ref{XX}) gives
\begin{equation}
-\sum_{\alpha=1}^N\frac{|x_\alpha|^2}{d\gamma_\alpha}
+\sum_{\alpha=1}^N\sum_{k=1}^{M_\alpha}\frac{|x_{\alpha,k}|^2}{a_{\alpha,k}}=c\Tr(X^\dagger X)-\frac{c}{d}|\Tr(X)|^2.
\end{equation}
Finally, eq. (\ref{lsum}) summed over $\beta=1,\ldots,L\leq N$ further simplifies to
\begin{equation}
\begin{split}
\sum_{\beta=1}^L\sum_{\ell=1}^{M_\beta}|\Tr(XR_{\beta,\ell})|^2
&=\sum_{\beta=1}^L\left[\frac{\gamma_\beta}{d}|\Tr(X)|^2-\frac{|x_\beta|^2}{d\gamma_\beta}
+\sum_{\ell=1}^{M_\beta}\frac{|x_{\beta,\ell}|^2}{a_{\beta,\ell}}\right]\\
&\leq 
\sum_{\beta=1}^L\frac{\gamma_\beta}{d}|\Tr(X)|^2+
\sum_{\beta=1}^N\left[-\frac{|x_\beta|^2}{d\gamma_\beta}
+\sum_{\ell=1}^{M_\beta}\frac{|x_{\beta,\ell}|^2}{a_{\beta,\ell}}\right]\\
&=\frac{\Gamma_L}{d}|\Tr(X)|^2+c\Tr(X^\dagger X)-\frac{c}{d}|\Tr(X)|^2\\
&=\frac{\Gamma_L-c}{d}|\Tr(X)|^2+c\Tr(X^\dagger X),
\end{split}
\end{equation}
where $\Gamma_L=\sum_{\alpha=1}^N\gamma_\alpha$. The equality is reached for $L=N$, and then
\begin{equation}
\sum_{\beta=1}^L\sum_{\ell=1}^{M_\beta}|\Tr(XR_{\beta,\ell})|^2
=\kappa_+|\Tr(X)|^2+\kappa_-\Tr(X^\dagger X).
\end{equation}
\end{proof}

This generalizes the results from refs. \cite{GEAM,GEAM_Pmaps,GEAM_kmaps} obtained for MU GETFs.
Now, consider a special case, where $X=\rho$ is the density operator and $P_{\alpha,k}=w_{\alpha,k}R_{\alpha,k}$ is the POVM from eq. (\ref{def3}).
Then, Proposition \ref{XIOC} allows us to find an upper bound for the partial index of coincidence
\begin{equation}
\mathcal{C}_L(\rho)=\sum_{\alpha=1}^N\sum_{k=1}^{M_\alpha}p_{\alpha,k}^2,\qquad
p_{\alpha,k}=\Tr(P_{\alpha,k}\rho).
\end{equation}
Indeed, eq. (\ref{sumL}) reduces to
\begin{equation}
\mathcal{C}_L(\rho)\leq\sum_{\alpha=1}^L\sum_{k=1}^{M_\alpha}\frac{p_{\alpha,k}^2}{a_{\alpha,k}}
\leq\frac{\Gamma_L-c}{d}+\kappa_-\Tr(\rho^2),
\end{equation}
where the first inequality follows from eq. (\ref{ca}), as $1=\sum_{\alpha=1}^N\sum_{k=1}^{M_\alpha}a_{\alpha,k}\geq a_{\beta,\ell}$. For $L=N$, one recovers the index of coincidence $\mathcal{C}(\rho)\equiv \mathcal{C}_N(\rho)$, and hence eq. (\ref{sumN}) gives
\begin{equation}
\mathcal{C}(\rho)\leq\kappa_++\kappa_-\Tr(\rho^2).
\end{equation}
Note that, unlike for the MU GETFs, conical 2-designs with unequal traces among $R_{\alpha,k}\in\mathcal{R}_\alpha$ never saturate the above inequality. This is the price on has to pay for less symmetric operators.

\subsection{Variance}

Observe that $R_{\alpha,k}$ with weights $w_{\alpha,k}=\Tr(R_{\alpha,k})$ provide a resolution of the rescaled identity operator,
\begin{equation}
\sum_{\alpha=1}^N\sum_{k=1}^{M_\alpha}w_{\alpha,k}R_{\alpha,k}=(d\kappa_++\kappa_-) \mathbb{I}_d.
\end{equation}
Interestingly, a similar resolution follows for $R_{\alpha,k}^2$.

\begin{Proposition}
If $\mathcal{R}$ is a conical 2-design, then the sum of its squared operators is proportional to the identity operator,
\begin{equation}\label{sum}
\sum_{\alpha=1}^N\sum_{k=1}^{M_\alpha}R_{\alpha,k}^2=
(\kappa_++d\kappa_-)\mathbb{I}_d.
\end{equation}
\end{Proposition}

\begin{proof}
Recall that if $\{G_j;\,j=0,\ldots,d^2-1\}$ is an orthonormal operator basis, then \cite{Werner_basis,Ohno}
\begin{equation}\label{GG}
\sum_{k=0}^{d^2-1}G_k^2=d\mathbb{I}_d.
\end{equation}
In our case, the informationally (over)complete set $\mathcal{R}$ is constructed from the Hermitian orthonormal operator basis $\{G_0=\mathbb{I}_d/\sqrt{d},\,G_{\alpha,k}:\,k=1,\ldots,M_\alpha-1;\,\alpha=1,\ldots,N\}$. Therefore, eq. (\ref{GG}) reduces to
\begin{equation}
\sum_{\alpha=1}^N\sum_{k=1}^{M_\alpha-1}G_{\alpha,k}^2=\frac{d^2-1}{d}\mathbb{I}_d.
\end{equation}
Now, let us assume that $R_{\alpha,k}$ are constructed from $G_{\alpha,k}$ via eq. (\ref{Gak}) (note that it does not matter which construction is chosen). Then,
\begin{equation}
\sum_{\alpha=1}^N\sum_{k=1}^{M_\alpha-1}R_{\alpha,k}^2=\sum_{\alpha=1}^N
\frac{1}{c_\alpha}\left(\sum_{k=1}^{M_\alpha}R_{\alpha,k}^2-\frac{\gamma_\alpha}{d} \mathbb{I}_d\right).
\end{equation}
This is a general formula for any MU GEWF and cannot be simplified further. However, if $R_{\alpha,k}$ form a conical 2-design, then $c_\alpha\equiv c$, and hence
\begin{equation}
\sum_{\alpha=1}^N\sum_{k=1}^{M_\alpha-1}R_{\alpha,k}^2
=\frac{(d^2-1)c+\Gamma}{d}\mathbb{I}_d.
\end{equation}
Using eq. (\ref{kappas}), we recover eq. (\ref{sum}).
\end{proof}

The above results are useful in calculating the total variance. It is one of the Brukner-Zeilinger invariants \cite{BZI1,BZI2}, defined as the sum of variances over an informatially complete orthogonal set of observables $\mathcal{A}=\{A_j:\,j=0,\ldots,d^2-1\}$ \cite{BZI3},
\begin{equation}\label{var}
\mathcal{V}(\mathcal{A},\rho)=\sum_{j=0}^{d^2-1}\left[\Tr(\rho A_j^2)-(\Tr\rho A_j)^2\right]=d-\Tr(\rho^2).
\end{equation}
Obviously, $\mathcal{V}(\mathcal{A},\rho)$ is invariant with respect to an arbitrary unitary transformation of the quantum state $\rho$. Here, we take an informationally overcomplete set of MU GEWFs that form a conical 2-design and show that an analogical property holds. Indeed, from definition,
\begin{equation}
\mathcal{V}(\mathcal{R},\rho)=\sum_{\alpha=1}^{N}\sum_{k=1}^{M_\alpha}
\left[\Tr(\rho R_{\alpha,k}^2)-(\Tr\rho R_{\alpha,k})^2\right]
=\kappa_-[d-\Tr(\rho^2)],
\end{equation}
where we used eqs. (\ref{sumN}) and (\ref{sum}).
Hence, the linear relation between the variance $\mathcal{V}(\mathcal{R},\rho)$ and purity $\Tr(\rho^2)$ is preserved. Actually, compared to eq. (\ref{var}), the right hand-side is rescaled by the positive factor $\kappa_-$, which multiplies the flip operator in the definition of conical 2-designs.

\subsection{Entanglement quantification}

For bipartite quantum states $\rho$ on $\mathbb{C}^d\otimes\mathbb{C}^d$, the Schmidt number $r$ measures the entanglement degree, from $r=1$ for separable to $r=d$ for maximally entangled states \cite{Sperling,TOPICAL}. The Schmidt number is defined as \cite{Terhal}
\begin{equation}
{\rm{SN}}(\rho)=\min_{\rho=\sum_jp_j|\psi_j\>\<\psi_j|}\max_j{\rm{SR}}(|\psi_j\>),
\end{equation}
where $\rm{SR}(|\psi_j\>)$ is the Schmidt rank of $|\psi_j\>$ \cite{Guhne} for the fixed decomposition $\rho=\sum_jp_j|\psi_j\>\<\psi_j|$. The Schmidt rank $r$ for a pure state $|\psi\>\in\mathbb{C}^d$ is recovered after reducing the state to the form
\begin{equation}
|\psi\>=\sum_{j=0}^{r-1}\lambda_j|e_j\>\otimes|f_j\>,
\end{equation}
where $\{|e_j\>;\,j=0,\ldots,d-1\}$, $\{|f_j\>;\,j=0,\ldots,d-1\}$ are local orthonormal bases, and $\lambda_j>0$ such that $\sum_{j=0}^{d-1}\lambda_j^2=1$ are the Schmidt coefficients.

Here, we obtain the Schmidt number criterion using the correlation matrix. From definition, the corelation matrix $\mathcal{B}(\rho)$ has the elements
\begin{equation}
\mathcal{B}_{\alpha,k;\beta,\ell}=\Tr[\rho(R_{\alpha,k}\otimes R_{\beta,\ell})]
\end{equation}
The entanglement measure of the density operator $\rho$ follows from the trace norm $\|\mathcal{B}(\rho)\|_{\tr}=\Tr\sqrt{\mathcal{B}^\dagger(\rho) \mathcal{B}(\rho)}$.

\begin{Proposition}
If the Schmidt number of a bipartite state $\rho$ is no greater than $r$, then
\begin{equation}
\|\mathcal{B}(\rho)\|_{\tr}\leq \kappa_++r\kappa_-.
\end{equation}
\end{Proposition}

The proof is identical to that of Theorem 5 in ref. \cite{GEAM_coherence}, because it is only based on the properties of conical 2-designs. From the convexity of the trace norm, it is enough to calculate $\|\mathcal{B}(\rho)\|_{\tr}$ for pure quantum states. The upper bound is obtained by using the property $\sum_{j=1}^r\lambda_j\leq\sqrt{r}$ \cite{Terhal} for the corresponding Schmidt coefficients. In particular, for $r=1$, we obtain the necessary separability criterion
\begin{equation}
\|\mathcal{B}(\rho)\|_{\tr}\leq \kappa_++\kappa_-.
\end{equation}
Our results generalize those obtained for SIC POVMs and mutually unbiased bases \cite{Morelli,ESIC}, general SIC POVMs \cite{Wangs,gESIC}, $(N,M)$-POVMs \cite{Schmidt_NM,SIC-MUB}, and generalized equiangular measurements \cite{GEAM_coherence}.

\section{Conclusions}

In this paper, we consider generalized equiangular weighted frames (GEWFs) with weights given by operator traces. We characterize these objects in terms of trace properties between operator pairs. From definition, it follows that every GEWF consists of linearly independent operators. If there are $M=d^2$ elements, then that GEWF is informationally complete. Next, we propose four construction methods from a single Hermitian orthonormal operator basis. For each method, we give an example in the maximal set of qubit operators. All this is then repeated for families of mutually unbiased (MU) GEWFs, which can form informationally overcomplete sets. Here, we also show relations between GEWF operators that come from different constructions. In what follows, general MU GEWFs are then compared with other important objects, such as mutually unbiased (MU) generalized equiangular tight frames (GETFs), conical 2-designs, and POVMs. It turns out that MU GEWFs are tight if and only if elements of every GEWF are of the same trace. Moreover, we establish an equivalence relation between subclasses of MU GEWFs and conical 2-designs. Finally, it turns out that any MU GEWFs can be used to define a POVM, but this POVM is in general not a conical 2-design. In terms of applications, we find the upper bound for the index of coincidence and prove a more general inequality for any linear operators. Then, we show that squared elements of MU GEWFs also provide a resolution of the identity operator, which in turn helps to simplify the formula for variance. We also use MU GEWFs to detect and quantify quantum entanglement. In particular, we formulate necessary conditions for the Schmidt number detection in terms of the correlation matrix.

In further research, it is interesting to analyze the properties of MU GEWFs in more depth, like conditions for their elements to all be projectors of a given rank. The next step is to establish their advantages for the applicational purposes in quantum tomography, quantum steering, entropic uncertainty relations, quantum coherence, quantum concurrence, and more.

\section{Acknowledgements}

This research was funded in whole or in part by the National Science Centre, Poland, Grant number 2025/58/E/ST2/00336. For the purpose of Open Access, the author has applied a CC-BY public copyright licence to any Author Accepted Manuscript (AAM) version arising from this submission.

\bibliography{C:/Users/cynda/OneDrive/Fizyka/bibliography}
\bibliographystyle{unsrt}

\end{document}